\documentclass[11pt]{article}

\usepackage[a4paper, margin=1in]{geometry}
\usepackage{amsmath,amssymb,amsthm,amsfonts,cite,xcolor,xspace,graphicx,xifthen,bm}
\usepackage{thmtools,thm-restate}
\usepackage{appendix}
\usepackage{enumitem}
\usepackage{bold-extra}
\usepackage[hang]{footmisc}
\usepackage{subcaption}
\setitemize{noitemsep,topsep=3pt,parsep=3pt,partopsep=3pt}

\definecolor{darkgreen}{rgb}{0,0.5,0}
\definecolor{darkblue}{rgb}{0,0,0.8}
\usepackage{hyperref}
\hypersetup{
    unicode=false,          
    colorlinks=true,        
    linkcolor=darkblue,          
    citecolor=darkgreen,        
    filecolor=magenta,      
    urlcolor=cyan,          
    hypertexnames=false
}

\usepackage[nameinlink,capitalize]{cleveref}

\newtheorem{theorem}{Theorem}[section]
\newtheorem{lemma}[theorem]{Lemma}
\newtheorem{corollary}[theorem]{Corollary}

\theoremstyle{definition}
\newtheorem{definition}{Definition}[section]
\theoremstyle{remark}
\newtheorem*{remark}{Remark}

\newcommand{\bigO}{\ensuremath{O}}
\newcommand{\pr}[2][]{\ifthenelse{\isempty{#1}}{\ensuremath{\Pr\left(#2\right)}}{\ensuremath{\Pr\left(#2 \mid #1\right)}}}
\newcommand{\NH}{\ensuremath{N}}

\newcommand{\LOCAL}{\ensuremath{\mathsf{LOCAL}}\xspace}
\newcommand{\SLOCAL}{\ensuremath{\mathsf{SLOCAL}}\xspace}

\newcommand{\PLOCAL}{\ensuremath{\mathsf{P}\text{-}\mathsf{LOCAL}}\xspace}
\newcommand{\PRLOCAL}{\ensuremath{\mathsf{P}\text{-}\mathsf{RLOCAL}}\xspace}
\newcommand{\PSLOCAL}{\ensuremath{\mathsf{P}\text{-}\mathsf{SLOCAL}}\xspace}

\DeclareMathOperator{\poly}{poly}
\DeclareMathOperator{\polylog}{polylog}
\DeclareMathOperator{\polyloglog}{polyloglog}
\newcommand{\eps}{\varepsilon}
\newcommand{\set}[1]{\left\{#1\right\}}

\newcommand{\minDegree}{\ensuremath{\delta}}
\newcommand{\rank}{\ensuremath{r}}

\newcommand{\hide}[1]{}

\begin{document}

\begin{flushleft}

\vspace*{0.8cm}
{\huge\bf On the Complexity of Distributed Splitting Problems}
\vspace{1.0cm}
\end{flushleft}

\newcommand{\auth}[3]{\textbf{#1}$\,\,\,\cdot\,\,\,$#2$\,\,\,\cdot\,\,\,$#3\par\medskip}

\auth{Philipp Bamberger}
{University of Freiburg}
{philipp.bamberger@cs.uni-freiburg.de}
\auth{Mohsen Ghaffari}
{ETH Z\"{u}rich}
{ghaffari@inf.ethz.ch}
\auth{Fabian Kuhn}
{University of Freiburg}
{kuhn@cs.uni-freiburg.de}
\auth{Yannic Maus}
{University of Freiburg}
{yannic.maus@cs.uni-freiburg.de}
\auth{Jara Uitto}
{University of Freiburg, ETH Z\"{u}rich}
{jara.uitto@inf.ethz.ch}

\vspace{1cm}

\begin{abstract}
One of the fundamental open problems in the area of distributed
graph algorithms is the question of whether randomization is needed
for efficient symmetry breaking. While there are fast,
$\poly\log n$-time randomized distributed algorithms for all of the
classic symmetry breaking problems, for many of them, the best
deterministic algorithms are almost exponentially slower. The following basic local splitting problem,
which is known as the \emph{weak splitting} problem takes a central
role in this context: Each node of a
graph $G=(V,E)$ has to be colored red or blue such that each node of
sufficiently large degree has at least one node of each color among
its neighbors. Ghaffari, Kuhn, and Maus [STOC '17] showed that this
seemingly simple problem is complete w.r.t.\ the above fundamental
open question in the following sense: If there is an efficient
$\poly\log n$-time determinstic distributed algorithm for weak
splitting, then there is such an algorithm for all locally checkable
graph problems for which an efficient randomized algorithm
exists. In this paper, we investigate the distributed complexity of
weak splitting and some closely related problems and we in
particular obtain the following results:
\begin{itemize}
	\item We obtain efficient algorithms for special cases of weak splitting, where the graph is nearly regular. In particular, we show that if $\delta$ and $\Delta$ are
	the minimum and maximum degrees of $G$ and if
	$\delta=\Omega(\log n)$, weak splitting can be solved
	deterministically in time
	$O\big(\frac{\Delta}{\delta}\cdot\poly(\log n)\big)$. Further, if
	$\delta = \Omega(\log\log n)$ and $\Delta\leq 2^{\eps\delta}$, there is a randomized algorithm with time
	complexity $O\big(\frac{\Delta}{\delta}\cdot\poly(\log\log n)\big)$.
	\item We prove that the following two related problems are also
	complete in the same sense: (I) Color the nodes of a graph with
	$C\leq \poly\log n$ colors such that each node with a sufficiently
	large polylogarithmic degree has at least $2\log n$ colors among
	its neighbors, and (II) Color the nodes with a large constant
	number of colors so that for each node of a sufficiently large at
	least logarithmic degree $d(v)$, the number of neighbors of each
	color is at most $(1-\eps)d(v)$ for some constant $\eps>0$.
\end{itemize}
\end{abstract}
\setcounter{page}{0}
\thispagestyle{empty}
\newpage

\section{Introduction}
In this paper, we investigate the distributed complexity of the \emph{splitting} problem and its variants in the \LOCAL model.\footnote{The \LOCAL model \cite{linial92,peleg00} is a standard synchronous message passing model on graphs, where in every round, each node can send an arbitrarily large message to each of its neighbors.} This problem is an important distributed symmetry breaking problem; to set the stage, let us start with an overview of the \emph{splitting} problem and its significance.

\subsection{The Splitting Problem and its Significance}
Splitting can be seen as a basic algorithmic tool to develop distributed divide-and-conquer algorithms for graph problems. Let us introduce it by using the well-studied vertex coloring problem as a toy example. Consider an $n$-node graph $G=(V, E)$ with maximum degree $\Delta$. Our objective is to color $V$ using as few colors as possible so that no two neighbors receive the same color. The best known deterministic distributed algorithm which is efficient---i.e., runs in $\poly\log n$ rounds--- computes a $\Delta \cdot 2^{O(\frac{\log \Delta}{\log\log \Delta})}$ coloring \cite{BE11}. Naturally, we would like to do much better; ideally $O(\Delta)$ or even just $\Delta+1$ colors, see, e.g., Open Problem 11.3 in the book by Barenboim and Elkin \cite{barenboimelkin_book}. 

Let us define the splitting problem to be dividing the nodes of the graph into two groups, say red and blue, such that the number of neighbors of each node in each group is at most $\frac{\Delta}{2}(1+\eps)$ for some small value $\eps$.\footnote{Something weaker would suffice for this special application; it would be enough if each node has at most $\frac{\Delta}{2}(1+\eps)$ neighbors in its own color. This is a form of defective coloring, and it is a weaker requirement than splitting. But let us use the convenient context of the coloring problem to motivate the stronger problem of splitting.}  If we had access to an efficient deterministic algorithm for splitting whenever $\Delta = \Omega(\log n/\eps^2)$, by repeated applications of it, we could partition the graph into $\frac{\Delta}{K}$ induced subgraphs, for $K \in \poly\log n$, each with maximum degree at most $K(1+\eps)^{\log \Delta}$. Thus, setting $\eps=o(1/\log \Delta)$, each subgraph would have degree $(1+o(1)) K$. Since we have efficient distributed algorithms for coloring graphs of maximum degree $d$ using $d+1$ colors in $\tilde{O}(\sqrt{d}) + O(\log^* n)$ rounds \cite{fraigniaud16}, we would immediately get a $\Delta(1+o(1))$-coloring for the whole graph in $\poly\log n$ rounds deterministically. This would be a breakthrough for the distributed coloring problem, and it would resolve a long-standing open problem. 

Of course, the catch is that we do not know an efficient deterministic method for constructing such a splitting. We emphasize that it is a matter of efficient construction and not a matter of existence. It is not hard to see that such a split always exists for $\Delta = \Omega(\log n/\eps^2)$, which is the regime where we need splitting, and in a randomized way it can be constructed (w.h.p.) by independently coloring each node red or blue uniformly at random. This nicely highlights the significance of splitting for distributed graph coloring: While there is a trivial randomized distributed algorithm for splitting that does not even require the nodes to communicate, an efficient deterministic algorithm would lead to major progress on the deterministic distributed coloring problem.

It is worth noting that the natural edge variant of the splitting problem is proved to be extremely instrumental for the variant of the coloring problem where we want to color the edges. Edge splitting (also known as degree splitting) can be defined as coloring all edges red or blue such that each node has at most $\frac{\Delta}{2}(1+\eps)$ edges in each color. Ghaffari and Su~\cite{GS17} provided a $\poly\log n$-round algorithm for edge-splitting, which led to the first efficient deterministic distributed $2\Delta(1+o(1))$-edge-coloring algorithm, thus partially resolving Open Problem 11.4 of \cite{barenboimelkin_book}. A significantly more efficient edge splitting algorithm was later provided in \cite{DISC17}. The most classic variant of the distributed edge coloring problem asks for a solution with $2\Delta-1$ colors as this is the number of colors obtained by a simple sequential greedy algorithm. The first efficient ($\poly\log n$ time) deterministic distributed algorithms for the $(2\Delta-1)$-edge coloring problem were obtained recently \cite{FischerGK17,GHK18}. These results were achieved by solving a generalization of the edge splitting problem in low-rank hypergraphs. Even more progress was achieved later and currently the best known efficient deterministic edge coloring algorithm---which is also based on solving edge splitting on the network graph and on some related low-rank hypergraphs---provides a $ (1+o(1))\Delta$-edge coloring~\cite{ghaffari2018edge-coloring}, whenever $\Delta=\Omega(\log n)$, thus almost matching the Vizing bound for the number of colors~\cite[Section 17.2]{bondy1976graph}.

Unfortunately, the splitting problem for vertices turned out to be much harder. Perhaps fortunately, it is also far more significant than just its relation to the coloring problem. It is tightly connected to the fundamental and long-standing open question of whether randomization is necessary for efficient distributed symmetry breaking. Currently, for many problems (such as $(\Delta+1)$-coloring or computing a maximal independent set (MIS)), there is an exponential gap between the best randomized algorithm and the best deterministic algorithms and whether $\poly\log n$-time deterministic algorithms for these problems exist is considered to be one of the main open problems in the area of local distributed graph algorithms~\cite{barenboimelkin_book}. Due to results of Ghaffari et al.~\cite{GKM17}, we now know the splitting problem is complete with respect to this question in the following sense. If one can find a $\poly\log n$-time deterministic distributed algorithm for splitting, then one can derandomize any $\poly\log n$-time randomized distributed algorithm for \emph{any} locally checkable problem into a $\poly\log n$-time deterministic distributed algorithm for that problem. The simple splitting problem (which has a trivial $0$-round randomized algorithm) therefore exactly captures the complete power of randomization for obtaining $\poly\log n$-time algorithms for local distributed graph problems.

In fact, Ghaffari et al.\cite{GKM17} showed that a much more relaxed version of the splitting problem is already complete: It is enough to ensure that each node with degree at least $\Omega(\log n)$ has at least one neighbor in each color. They call this the \emph{weak splitting} problem and they showed that if one can find a $\poly\log n$-time deterministic distributed algorithm for weak splitting, that also implies that one can derandomize any $\poly\log n$-time randomized distributed algorithm for \emph{any} locally checkable problem into a $\poly\log n$-time deterministic distributed algorithm for that problem. Notice that weak splitting would not be sufficient for the method described above for the coloring problem. The proof of Ghaffari et al.\cite{GKM17} for using weak splitting goes a very different route, it uses weak splitting to build a certain network decomposition, and \cite{GHK18} shows how to use such network decompositions to derandomize randomized algorithms for any locally checkable problem.

To summarize, weak splitting---which might even look deceivingly simple---is all that we need so that we can obtain deterministic $\poly\log n$-time algorithms for locally checkable graph problems and thus to answer many of the outstanding open questions regarding efficient deterministic local graph algorithms. In this paper, we show some partial progress on our understanding of the weak splitting problem, and we also show that even some very relaxed variants of it can be proven to be complete (in the same completeness sense as weak splitting).

\subsection{Our Contributions}
\label{sec:contributions}

\paragraph{Algorithmic Results/Deterministic:} Our algorithmic contribution is a new weak splitting algorithm that is efficient in nearly regular graphs. Let $\Delta$ and $\delta$ denote the maximum and minimum degrees of the given graph. If $\delta=\Omega(\log n)$, we give a deterministic algorithm that solves weak splitting in $O(\frac{\Delta}{\delta} \poly\log n)$ rounds. Hence, for all graphs that are somewhat regular and satisfy $\frac{\Delta}{\delta} \leq \poly\log n$, we obtain a deterministic $\poly\log n$-time weak splitting algorithm. To state the result more formally, and to open way for our other results, let us phrase the splitting problem in a more general format.

We first introduce some notation and terminology. Let us consider a bipartite graph $B=(U\cup V, E)$ where we view nodes of $U$ as \emph{constraint nodes} and nodes of $V$ as \emph{variable nodes}. Equivalently, we can think of $U$ as vertices of a hypergraph and $V$ as the hyperedges of it. Throughout the paper, when using a bipartite graph $B=(U\cup V, E)$, we refer to $U$ as the left side, $V$ as the right side, and we use $\delta_B$ and $\Delta_B$ to denote the minimum and maximum degree of nodes in $U$ and we use $\rank_B$ to denote the maximum degree of nodes in $V$, where $\rank_B$ stands for the rank of the corresponding hypergraph, i.e., the maximum number of vertices in a hyperedge. We omit the subscripts if the corresponding graph is clear from the context.  Weak splitting can then formally be defined as follows:

\begin{definition}[Weak Splitting]\label{def:weaksplit}
Let $B=(U\cup V,E)$ be a bipartite graph. A weak splitting of $B$ is a $2$-coloring of the nodes in $V$ such that every node in $U$ has at least one neighbor of each color.
\end{definition}

Notice that the splitting problem on general graphs $G=(V_G,E_G)$ discussed above can be phrased as such a bipartite/hypergraph problem: for each node $v\in V_G$, make two copies of it, one for $v_L\in U$ and one for $v_R\in V$. For each edge $\{u, v\}\in E_G$, we connect $v_L$ to $u_R$ and $v_R$ to $u_L$. Distinguishing these left $U$ and right $V$ sides allows us to distinguish between the constraints and the variables of the problem, and facilitates our discussions in several places. We prove the following.

\begin{theorem}\label{thm:detsplitting}
There is a deterministic distributed algorithm that in any $n$-node bipartite graph $B=(U\cup V, E)$ in which the minimum degree of the nodes in $U$ is $\delta \geq 2\log n$, solves the weak splitting problem in \smash{$O\big(\frac{r}{\delta}\cdot\poly\log n\big)$} rounds.
\end{theorem}

In addition, in Theorem \ref{thm:DeltaLargersixr}, we show that if $\delta\geq 6\rank$, the above problem can even be solved in $\poly\log n$ time without any additional requirement on $\delta$ (i.e., without requiring that $\delta=\Omega(\log n)$). However, this result can not be applied to general graphs as converting a graph to a bipartite graph as described above will always yield a bipartite graph with $\delta\leq\rank$.

Further, in Theorem \ref{thm:girth10det}, we prove that if the bipartite graph $B$ has girth at least $10$, the requirement on $\delta$ can be improved to $\delta=\Omega(\sqrt{\log n})$.

The above results provide only a partial progress on our understanding of the weak splitting problem and they certainly fall short of the ultimate goal of enabling us to derandomize any $\poly\log n$-round randomized algorithm for any locally checkable problem to a $\poly\log n$-round deterministic algorithm for it. If we could strengthen the above result in one of two directions, that would be a breakthrough: (A) If we could extend this weak splitting to all graphs, we would get the aforementioned desired derandomization algorithm, thus resolving many classic open problems of distributed graph algorithms, including the first three in the Open Problems section of the book of Barenboim and Elkin~\cite{barenboimelkin_book};  (B) Alternatively, if we could change this weak splitting algorithm for nearly regular graphs to a splitting algorithm for nearly regular graphs, then we would obtain a $\Delta(1+o(1))$-coloring algorithm in $\poly\log n$ rounds, hence resolving Open Problem 11.3 of \cite{barenboimelkin_book}. We think that the partial progress that this paper provides may still be a concrete step in approaching these ultimate goals. 

\vspace*{-3mm}
\paragraph{Algorithmic Results/Randomized:} In addition, we study randomized algorithms for the weak splitting problem. The randomized complexity of the problem might be interesting in the context of the recent interest in understanding the complexity landscape of randomized sublogarithmic-time distributed graph algorithms. In \cite{chang2017time}, Chang and Pettie show that for any locally checkable labeling (LCL) problem\footnote{LCL problems \cite{naor95} are graph problems, where the output of every node is a label from a bounded alphabet and the validity of a solution can be checked a deterministic constant-time distributed algorithm.} for which a randomized algorithm with running time $o(\log_{\Delta}n)$ exists, this algorithm can be sped up to run in the time for solving generic instances of problems where the existence of a solution follows from a polynomially relaxed version of the Lov\'{a}sz Local Lemma (LLL). The best known generic randomized algorithm for such LLL problems on bounded-degree graphs---and in fact also for graphs with $\poly\log\log n$ degrees---is $2^{o(\sqrt{\log^\eps\log n})}$~\cite{ghaffari-lll,GHK18}, for any constant $\eps>0$. Moreover, it is conjectured in \cite{chang2017time} that this complexity should be $\poly\log\log n$ or even $O(\log\log n)$. The weak splitting problem---and also the splitting problem more generally---is a particularly simple and seemingly well-behaved problem that falls into this class of LLL problems and it would therefore be interesting to understand whether at least weak splitting can be solved in time $\poly\log\log n$ in bounded-degree graphs or in graphs with degrees at most $\poly\log\log n$. We make some partial progress and prove that this is at least true for bipartite graphs $B=(U\cup V, E)$, where the minimum degree $\delta$ in $U$ is at least $\Omega(\log\log n)$.

\begin{theorem} \label{thm:mainweaksplit} 
Consider an arbitrary $n$-node bipartite graph $B=(U\cup V, E)$ where the minimum degree in $U$ is $\delta\geq c\log(\rank\log n)$ for a sufficiently large constant $c>1$. Then, there is a randomized distributed algorithm that in $O\big(\frac{\rank}{\delta}\cdot\poly(\log (\rank\log n))\big)$ rounds solves the weak splitting problem on $B$.
\end{theorem}

Similar to the case of deterministic algorithms, we show that slightly stronger results hold for special cases, also for randomized algorithms. As long as $\delta\geq6\rank$, the problem can always be solved in time $\poly\log\log n$ (Theorem \ref{thm:DeltaLargersixr}) and if the bipartite graph has girth at least $10$, the problem can be solved in time $O\big(\frac{\rank}{\delta}\cdot\poly\log\log n\big)$ even if we only require that $\delta=\Omega(\sqrt{\log\log n})$ (Theorem \ref{thm:girth10rand}).

\vspace*{-3mm}
\paragraph{Hardness Results:} To strengthen our understanding of the splitting problem, we also investigate it from the (conditional) hardness side. Our goal here is to identify weaker and alternative forms of splitting, which are still complete in the above sense. Let us first briefly introduce the necessary formal background. Let \PLOCAL and \PRLOCAL be the classes of $\poly\log(n)$-\emph{locally checkable}\footnote{A graph problem is called $d$-locally checkable if given a solution, there is a deterministic $d$-round \LOCAL algorithm in which every node outputs ``yes'' if and only if the given solution if a valid solution for the given problem \cite{fraigniaud13}. LCL problems are a special important case of $O(1)$-locally checkable problems.}  graph problems that can be solved by $\poly\log(n)$-time deterministic and $\poly\log(n)$-time randomized \LOCAL algorithms, respectively. We say that a graph problem $\mathcal{P}$ is \emph{\PRLOCAL-complete} if it is in \PRLOCAL and if a $\poly\log(n)$-time deterministic \LOCAL algorithm for $\mathcal{P}$ would imply that $\PLOCAL=\PRLOCAL$. In \cite{GKM17}, Ghaffari, Kuhn, and Maus show that the weak splitting problem on bipartite graphs $B=(U\cup V,E)$ is \PRLOCAL-complete even if the minimum degree in $U$ is at least polylogarithmic in $n$.\footnote{To be precise, in \cite{GKM17}, completeness was not shown w.r.t.\ to \PRLOCAL, but w.r.t.\ a class \PSLOCAL of efficient deterministic \emph{sequential} local algorithms. However, as shown in \cite{GHK18}, for $\poly\log(n)$-locally checkable problems, these two classes coincide.} We define the following two multicolor variants of the splitting problem, which are much more relaxed, and we show that they are still complete. As splitting, both problems are defined on a bipartite graph $B=(U\cup V,E)$.

\begin{definition}[$(C,\lambda)$-Multicolor Splitting] 
Given a bipartite graph $B=(U\cup V,E)$ and parameters $C\geq 2$ and $\lambda\geq  2/C$, a \textit{$(C,\lambda)$-multicolor splitting} of $B$ is a coloring of the nodes in $V$ with $C$ colors such that each $u\in U$ of degree has at most $\lceil\lambda\cdot\deg(u)\rceil$ neighbors of each color.
\end{definition}

\begin{definition}[$C$-Weak Multicolor Splitting] 
Given a bipartite Graph $B=(U\cup V, E)$, a \textit{$C$-weak multicolor splitting} of $B$ is a coloring of the nodes in $V$ with $C\geq 2\log n$ colors such that each node $u\in U$ of degree $\deg(u)\geq 2(\log n+1)\ln n$ sees at least $2\log n$ different colors.
\end{definition}

We note that here (and throughout the rest of the paper), we use $\log x$ to refer to $\log_2 x$ and we use $\ln x$ to refer to the natural logarithm.  In \Cref{thm:weakmulticolorsplit}, we show that the $C$-weak multicolor splitting problem is \PRLOCAL-complete for any $C\leq\poly\log n$, even if the minimum degree of nodes in $U$ is $\log^c n$ for an arbitrary constant $c\geq 2$. Further, in \Cref{thm:multicolorsplit}, we prove that $(C,\lambda)$-multicolor splitting is \PRLOCAL-complete as long as the minimum degree of nodes in $U$ is at least $\alpha\ln^2 n$ for a sufficiently large constant and as long as $C\leq \poly\log n$ and $\lambda \leq C^{-\eps}$ for some constant $\eps>0$.

\vspace*{-3mm}
\paragraph{Parameter Preserving Hardness Results:} The hardness results mentioned above, and also those of \cite{GKM17}, provide reductions from graph or hypergraph problems where the degrees might become very large, up to $n$, even if the degree in the original graph was somewhat small. Our last contribution is to provide some degree preserving reductions from the problems of maximal independent set and $\Delta(1+o(1))$-coloring to the splitting problem. In \Cref{sec:parameter-preserving-reductions}, we show that if the latter can be solved in time $T_{n, \Delta}$ in graphs of degree $\Delta$ and $n$ nodes, then the former problems can also be solved in $\poly(\log n) \cdot T_{n, \Delta}$. Hence, for instance, an algorithm with complexity $\Delta^{1-\eps} \poly(\log n)$ for the splitting problem, for any constant $\eps>0$, would yield an MIS algorithm with complexity $\Delta^{1-\eps} \poly(\log n)$, which would be better than all known algorithms whenever $\Delta \in (\poly(\log n), 2^{O(\sqrt{\log n})})$.

\hide{

\ym{Need something that says what an instance is? LLL solvability?}

}


\vspace*{-2mm}
\section{Algorithms for Weak Splitting}
\label{sec:weaksplitting}
\vspace*{-2mm}

We next describe our deterministic and randomized algorithms for the weak splitting problem.

\vspace*{-2mm}
\subsection{Basic Deterministic Weak Splitting Algorithm}
\label{sec:derand}
\vspace*{-2mm}
If the minimum degree $\delta$ of the nodes on the left hand side is at least $2\log n$, a union bound shows that the following simple randomized algorithm solves the weak splitting problem w.h.p. 
\begin{center}\textit{Color each node on the right hand side red$/$blue with probability $1/2$ each.}\end{center}
Using the derandomization results from \cite{GHK18} this can be derandomized given a suitable coloring of the input graph. The formal result is as follows. 
\begin{lemma}\label{lem:Easyderand}
There is a deterministic algorithm to compute a weak splitting in time $O(\Delta\cdot \rank)$ if $\delta\geq2\log n$.
\end{lemma}
\begin{proof}
Let $B=(U\cup V, E)$ be an instance of the weak splitting problem with minimum degree $\delta\geq 2\log n$ on the left hand side. 
In the aforementioned randomized algorithm the probability that some $u\in U$ has a monochromatic neighborhood is
\[\pr{\text{all neighbors are red}}+\pr{\text{all neighbors are blue}}=\frac{2}{2^{\text{deg}(u)}}\leq\frac{2}{2^\delta}\leq\frac{2}{n^2}~.\]
With a union bound over all nodes in $U$ we obtain that the probability that there is a node with a monochromatic neighborhood is at most $2/n<1$ and a node can check whether  it has a monochromatic neighborhood by looking at its $1$-hop neighborhood. 
Hence, by \cite[Theorem III.1]{GHK18}, this randomized $0$-round algorithm with checking radius $1$ can be derandomized into an $\SLOCAL(2)$-algorithm. By \cite[Proposition 3.2]{GHKfull} this can be transformed into an $O(C)$ \LOCAL algorithm if a $C$-coloring of $B^2$, i.e., the graph that one obtains from $B$ by additionally connected any two nodes in distance at most two to each other, is  given. 
As the maximum degree of $B^2$ is $\Delta\rank$ we can compute the necessary coloring with $\bigO(\Delta\rank)$ colors and in $\bigO(\Delta\rank+\log^*n)$ rounds, e.g., with the algorithm from \cite{BEK14}. Thus the total runtime can be bounded as
$\bigO(\Delta\rank+\log^*n)=O(\Delta\rank)$ as $\Delta\geq \delta\geq 2\log n$.
\end{proof}

\subsection{Deterministic Degree-Rank Reduction}
\label{sec:degreeRankI}

\begin{lemma}\label{lem:derand}
There is a deterministic algorithm to compute a weak splitting in time $O(\rank\cdot\log n)$ if $\delta\geq2\log n$.
\end{lemma}
\begin{proof}
Let $B=(U\cup V, E)$ be an instance of the weak splitting problem with minimum degree $\delta\geq 2\log n$ on the left hand side. 
If $\delta>2\log n$, each node in $U$ deletes an arbitrary set of its incident edges such that at least $\delta'=\lceil2\log n\rceil$ remain. By \Cref{lem:Easyderand}, we can compute a weak splitting on the resulting graph $H$ in $\bigO(\delta'\cdot\rank)=\bigO(\rank\cdot\log n)$ rounds. The computed coloring of the right hand side of $H$ immediately induces a weak splitting of the original graph $B$ as the weak splitting property is conserved under adding edges to a graph.
\end{proof}

In the algorithm in \Cref{lem:derand}, we deleted edges of high degree nodes in $U$ arbitrarily. The idea of our main deterministic weak splitting algorithm is to do this deletion more thoughtfully such that we are guaranteed that also the rank shrinks to a sufficient extent.

\begin{definition}
Given an undirected (multi-)graph $G=(V,E)$, a \textit{directed degree splitting} of $G$ with discrepancy $\kappa:\mathbb{N}\to\mathbb{R}$ is an orientation of the edges of $G$ such that for every node $v\in V$, the absolute value of the difference between the number of its incident incoming and its incident outgoing edges is at most $\kappa(\deg_G(v))$.
\end{definition}

In \cite{DISC17} it was shown that directed degree splittings can be computed efficiently.
\begin{theorem}[Theorem 1 in \cite{DISC17}]\label{thm:DISC17}
For every $\varepsilon>0$, there is a deterministic $\bigO(\varepsilon^{-1}\cdot\log\varepsilon^{-1}\cdot(\log\log\varepsilon^{-1})^{1.71}\cdot\log n)$ round distributed algorithm for directed degree splitting such that the discrepancy at each node $v$ of degree $d(v)$ is at most $\varepsilon\cdot d(v)+2$.
The randomized runtime of the same result is $\bigO(\varepsilon^{-1}\cdot\log\varepsilon^{-1}\cdot(\log\log\varepsilon^{-1})^{1.71}\cdot\log\log n)$
\end{theorem}
Note that the randomized runtime of the theorem is not stated in \cite{DISC17} but follows by substituting each deterministic $O(\log n)$-round sinkless orientation algorithm in their proofs with the randomized $O(\log\log n)$-round sinkless orientation algorithm from \cite{GS17}.
To ease presentation, we omit the $\log\log$ term and upper bound the runtime of the directed degree splitting algorithm by $\bigO(\varepsilon^{-1}\cdot(\log\varepsilon^{-1})^{1.1}\cdot\log n)$ whenever we apply \Cref{thm:DISC17}.
We now iteratively use the degree splitting algorithm from \Cref{thm:DISC17} to reduce the degrees on the left hand side and the rank on the right hand side at the same time. 

\smallskip

\textbf{Degree-Rank Reduction I:} Given a bipartite graph $B=(U\cup V,E)$ and parameters $\varepsilon$ and $k$, compute a directed degree splitting on $B$ with discrepancy at most $\varepsilon\cdot d(v)+2$ for each $v\in U\cup V$. Now that all edges are oriented, delete all edges from $B$ that are directed from a node in $V$ towards a node in $U$. Repeat this process on the residual graph. Stop after $k$ iterations.

We can lower bound the degree on the left hand side and upper bound the 'rank' on the right hand side after $k$ iterations of the algorithm as follows.
\begin{lemma}\label{lem:degShrink}
Let $B$ be a bipartite graph with minimum degree $\delta$ and rank $\rank$ and let $\delta_k$ ($\rank_k$) be the minimum degree (rank) of the graph obtained after $k$ iterations of the Degree-Rank Reduction Algorithm on $B$ with some $0<\varepsilon<1/3$. Then
\begin{align*}
 \delta_k& >\left(\frac{1-\varepsilon}{2}\right)^k\delta-2 & \text{~and~} &
& \rank_k & <\left(\frac{1+\varepsilon}{2}\right)^k\rank+3~.
\end{align*}
\end{lemma}

\begin{proof}
In each iteration only incoming edges to nodes in $V$ survive. If a node has $\delta_i$ edges before iteration $i$ it has at least $\frac{1-\eps}{2}\delta_i-1$ incoming edges in the directed splitting computed in iteration $i$. Thus a simple induction shows that after $k$ iterations the minimum degree of nodes in $V$ can be lower bounded by 
 \[\delta_k\geq\left(\frac{1-\varepsilon}{2}\right)^k\delta-\sum\limits_{i=0}^{k-1}\left(\frac{1-\varepsilon}{2}\right)^i~.\]
This implies the first claim as $\sum\limits_{i=0}^{k-1}\left(\frac{1-\varepsilon}{2}\right)^i\leq\sum\limits_{i=0}^{k-1}\left(\frac{1}{2}\right)^i<2$~.

Similar to the first claim one can show by induction that the maximum degree on the right hand side can be upper bounded by  \[\rank_k\leq\left(\frac{1+\varepsilon}{2}\right)^k\rank+\sum\limits_{i=0}^{k-1}\left(\frac{1+\varepsilon}{2}\right)^i~.\]
This implies the second claim as $\sum\limits_{i=0}^{k-1}\left(\frac{1+\varepsilon}{2}\right)^i<\sum\limits_{i=0}^{\infty}\left(\frac{1+\varepsilon}{2}\right)^i=\left(1-\frac{1+\varepsilon}{2}\right)^{-1}<3\text{\quad for }\varepsilon<1/3$~.
\end{proof}

We can now prove \Cref{thm:detsplitting}, our main deterministic splitting result. The following is a more precise version of  \Cref{thm:detsplitting}.

\begin{theorem}\label{thm:weakSplittingSpeedup}
There is a deterministic distributed algorithm that, given a weak splitting instance $B=(U\cup V, E)$ with $\minDegree\geq2\log n$, computes a weak splitting in time\[\bigO\left(\frac{\rank}{\delta}\cdot\log^2n+\log^3n\left(\log\log n\right)^{1.1}\right)~.\]
\end{theorem}

\begin{proof}
Let $B=(U\cup V, E)$ a weak splitting  instance. If $\delta\leq48\log n$, the algorithm from \Cref{lem:derand} gives an $\bigO(\rank\cdot\log n)=\bigO(\rank/\delta\cdot\log^2n)$ algorithm. Thus for the rest of the proof we assume that $\delta>48\log n$.
Let $\varepsilon=\min\{1/k,1/3\}$ and  $k:=\lfloor\log\left(\frac{\delta}{12\log n}\right)\rfloor$. Let $\bar{B}$ denote the bipartite graph that we obtain after $k$ iterations of the degree-rank reduction algorithm with accuracy $\eps$. 
Due to \Cref{lem:degShrink} the maximum rank of $\bar{B}$ can be upper bounded as 
\begin{align*}
\rank_{\bar{B}}<\left(\frac{1+1/k}{2}\right)^k\rank+3\leq\left(\frac{e^{1/k}}{2}\right)^k\rank+3= \frac{e}{2^k}\cdot \rank+3\leq 24e\cdot\frac{\rank}{\delta}\log n+3
\end{align*}
and the minimum degree of the nodes on the left hand side can be lower bounded by 
\begin{align*}
\delta_{\bar{B}}\stackrel{(*)}{>}\left(\frac{1-1/k}{2}\right)^k\cdot \delta-2\geq   \frac{1}{4}\cdot 2^{-k}\cdot \delta-2\geq 12\log n-2\stackrel{(n\geq 4)}{\geq}2\log n~.
\end{align*}
At $(*)$ we used that $\delta>48\log n$ implies that we have more than two iterations of the splitting algorithm, i.e., $k>2$ which implies $\left(1-1/k\right)^k\geq1/4$.
Now, we use \Cref{lem:derand} to compute a weak splitting on $\bar{B}$ that is also a weak splitting on the original graph $B$.
The runtime of computing a weak splitting on $\bar{B}$ is $\bigO(\rank_{\bar{B}}\cdot\log n)=\bigO(\rank/\delta\cdot\log^2n)$. 
The runtime of each of the $k=O(\log n)$ execution of the degree-rank reduction algorithm is $\bigO(\varepsilon^{-1}\cdot(\log\varepsilon^{-1})^{1.1}\cdot\log n)$ (cf. \Cref{thm:DISC17}). Due to $\varepsilon^{-1}=\lfloor\log\left(\frac{\delta}{12\log n}\right)\rfloor$ the time complexity of all iterations can be bounded by $\bigO\left(\log^2n(\log\log n)^{1.1}\right)$ and the total runtime of the algorithm is $\bigO(\rank/\delta\cdot\log^2n+\log^3n(\log\log n)^{1.1})$.
\end{proof}

\subsection{An Efficient Deterministic Algorithm when \boldmath$\delta\geq6\rank$}
The splitting algorithm that we use in the degree-rank reduction I in \Cref{sec:degreeRankI} has an inaccuracy on both sides of the bipartite graph; in particular, a node on the right hand side that has $2$ or less edges remaining might loose all of its incident edges in one iteration. To solve the weak splitting problem efficiently for $\delta\geq 6\rank$ we define a degree-rank reduction algorithm that always obtains discrepancy one or zero on the right hand side of the bipartite graph. 

\smallskip

\textbf{Degree-Rank Reduction II:} Given a bipartite graph $B=(U\cup V,E)$ and an \emph{accuracy parameter} $\varepsilon$ we define one iteration of the {degree-rank reduction II} as follows: Each $v\in V$ groups its $d:=\deg(v)$ neighbors  $u_1,\dots,u_d\in N_B(v)$ into pairs $(u_1,u_2),(u_3,u_4),\dots$ (if $d$ is odd, $u_d$ remains unpaired). Then we define a multigraph $G$ with vertex set $U$. $G$ contains an edge $e=\{u_i,u_{i+1}\}$ for any of these  pairs and we say that $v$ is the \emph{corresponding node} for edge $e$. Note that there can be multiple edges between two nodes in $G$ with distinct corresponding nodes.
 Now, we compute a directed degree splitting on $G$ with discrepancy at most $\varepsilon\cdot\deg_G(u)+2$ for each $u\in U$. 
We obtain a residual graph $B'\subseteq B$ through removing edges from $B$ as follows: For any edge $e=\{u,\bar{u}\}$ of $G$ let $v_e$ be the corresponding node for $e$. If $e$ is directed from $u$ to $\bar{u}$, delete the edge $\{\bar{u},v_e\}$ from $B$; if $e$ is directed from $\bar{u}$ to $u$, delete the edge $\{u,v_e\}$ from $B$. All other edges of $B$ remain and are edges of the residual graph $B'$. 
If we consider several iterations of the degree-rank reduction we always repeat the process on the residual bipartite graph.

\medskip

The crucial property of the above algorithm is that the rank of the bipartite graph can never go below one as any node on the right hand side keeps at least one out of two neighbors and if it has only one neighbor left it also keeps that one.

\begin{lemma}\label{lem:degShrink2}
Let $B$ be a bipartite graph with rank $\rank$ and let $\rank_k$ be the minimum degree (rank) of the graph obtained after $k$ iterations of the Degree-Rank Reduction II Algorithm on $B$ applied with an arbitrary accuracy parameter $\eps$. Then we have $\rank_{\lceil\log\rank\rceil}=1$.
\end{lemma}
\begin{proof}
First, we perform an induction over the number of iterations and show that 
\begin{align}\label{eqn:rankexact}
\rank_k<\frac{1}{2^k}\rank+\sum\limits_{i=1}^k\frac{1}{2^i}~
\end{align} holds. 
For $k=0$ the hypothesis (\ref{eqn:rankexact}) simplifies to $\rank<\rank+1$ and is trivially satisfied. 

\noindent\textit{Induction Step:} Assume the that (\ref{eqn:rankexact}) holds for $k$. As a node on the right hand side never looses more than half of its edges in one iteration of rank-reduction II we have  $\rank_{k+1}=\lceil\frac{\rank_k}{2}\rceil$. 
As $\rank_k$ is an integer we obtain
 \[\rank_{k+1}\leq \frac{\rank_k}{2}+\frac{1}{2}\stackrel{\text{I.H.}}{<}\frac{1}{2^{k+1}}\rank+\sum\limits_{i=2}^{k+1}\frac{1}{2^i}+\frac{1}{2}=\frac{1}{2^{k+1}}\rank+\sum\limits_{i=1}^{k+1}\frac{1}{2^i}~.\]
With $\sum\limits_{i=1}^k\frac{1}{2^i}<1$ we obtain $\rank_k<\frac{1}{2^k}\rank+1$ for any $k$ and in particular $\rank_{\lceil\log\rank\rceil}<2$. Hence we obtain $\rank_{\lceil\log\rank\rceil}=1$ as $\rank_{\lceil\log\rank\rceil}$ is an integer and cannot be smaller than $1$.
\end{proof}

The following theorem is obtained by using the degree-rank reduction II for $\lceil\log r\rceil$ iterations until the rank $r$ of the remaining graph is $1$. With the condition $\delta\geq 6\rank$, one can then show that the minimum degree of the left-side nodes is still at least $2$ and we have thus reduced the problem to a trivial weak splitting instance.

\begin{theorem}\label{thm:DeltaLargersixr}
If $\delta\geq6\rank$, we can solve the weak splitting problem deterministically in $\polylog n$ rounds and randomized in $\polyloglog n$ rounds.
\end{theorem}
\begin{proof}
If $\delta\geq 2\log n$ we can solve the problem deterministically with the algorithm from \Cref{thm:weakSplittingSpeedup}  in $\bigO(\log^2n+\log^3n(\log\log n)^{1.1})$ rounds.
Thus assume that  $\delta<2\log n$. Set $\eps=1/(10\Delta)=1/(20\delta)$ (see the comment at the beginning of \Cref{sec:rand} which states that it is sufficient to solve weak splitting with almost regular degrees on the left hand side to solve it for all degrees) and execute $k=\lceil \log r\rceil$ iterations of degree-rank reduction II.  We now want to lower bound the minimum degree after these $k$ iterations.
As $\eps\cdot d(u)<1$ for all nodes $u\in U$ and $\deg_G(u)\leq \deg_B(u)$ we obtain that the discrepancy of the computed splitting in degree-rank reduction is at most $1$ if the degree of $u$ is odd and $2$ if the degree of $u$ is even. Thus a node $u\in U$ with initial degree $\delta$ has degree at least $\delta/2-1$ after one iteration of the algorithm. If $r\leq 2$ we only need one iteration and obtain that the minimum degree after this iteration is at least $\delta/2-1\geq 2$. For $r>2$ an induction over the number of iterations shows that the minimum degree in iteration $\lceil \log r\rceil$ is strictly larger than $\delta\cdot 2^{-\lceil\log r\rceil}-2$. As we have
$2^{\lceil \log r \rceil}\leq 2r-2$ for $r>2$ we obtain that the minimum degree after $k=\lceil\log r\rceil$ iterations is strictly larger than 
\begin{align*}
\frac{\delta}{2^{-\lceil\log \rank\rceil}}-2\geq \frac{\delta}{2\rank-2}-2\stackrel{(\delta\geq 6\rank)}\geq \frac{\delta}{\frac{\delta}{3}-2}-2=\frac{3}{1-\frac{6}{\delta}}-2\stackrel{(\delta\geq 6r\geq 12)}=1~,
\end{align*}
 i.e., the mininmum degree is at least $2$. Thus in all cases the resulting graph has rank $1$ and minimum degree at least $2$. Thus, every node in $U$ can choose one of its neighbor to be colored red and one of its neighbors to be colored blue. The obtained coloring is a weak splitting of $B$.
The deterministic runtime of the algorithm with $\delta<2\log n$ is $O(k\cdot \varepsilon^{-1}\cdot(\log\varepsilon^{-1})^{1.1}\cdot\log n)=O(\log^3 n \log^{1.1}\log n)$ due to the deterministic runtime of \Cref{thm:DISC17}. 

The proof of the $\polylog\log n$ randomized algorithm is along similar lines. If $\delta\geq 2\log n$ the zero round randomized algorithm that we explain at the beginning of \Cref{sec:derand} solves the problem.  If $\delta,\rank\leq 2\log n$ but 
$\delta\geq c \log\log n$ for a sufficiently large constant $c$ we use \Cref{thm:mainweaksplit} to solve the problem in $\frac{\rank}{\delta}\polyloglog n=\polyloglog n$ rounds---note that $\delta\geq c\log\log n$ implies $\delta\geq c' \log (\rank \log n)$ for a slightly smaller constant $c$ if $\rank\leq\delta\leq 2\log n$. If $\delta<c\log\log n$ we use the degree-reduction II as in the deterministic case. The runtime bound follows with $\eps^{-1}=\polyloglog n$ and the randomized runtime in  \Cref{thm:DISC17}.
\end{proof}

\subsection{A Randomized Algorithm for Weak Splitting}
\label{sec:rand}

For proving \Cref{thm:mainweaksplit} we assume \emph{almost uniform} degrees on the left hand side, i.e., $\delta>\Delta/2$. This is sufficient because we can split each node $u\in U$ with $\deg(u)>2\delta$ into $\lfloor\deg(u)/\delta\rfloor$ virtual nodes with degree at least $\delta$ and less than $2\delta$ and compute a weak splitting on this graph, which directly induces a weak splitting on the original graph.
The algorithm of \Cref{thm:mainweaksplit} is based on the infamous graph shattering technique that consists of two parts: In the first part, we use a random process to fix the colors of some nodes in $V$ such that the probability that a node in $U$ is \emph{unsatisfied}, i.e., does not have a red and blue neighbor, is $1/\poly\Delta$. The residual graph consisting of the unsatisfied nodes on the left side and the uncolored nodes on the right side will have small connected components, say size $\poly(\log n)$, such that we can efficiently compute a weak splitting with the deterministic algorithm from \Cref{thm:weakSplittingSpeedup}. For more information on the shattering technique consult, e.g., \cite{BEPS16}.
 More formally we use the following theorem to bound the size of the unsolved components after the first part. 

\begin{theorem}[Theorem V.1 in \cite{GHK18}]\label{thm:shattering}
Suppose a random process in which given a graph $G$, each node $v$ survives to a residual graph $H$ with probability at most $(e\Delta)^{-4c}$ and this bound holds even for adversarial choices of the random bits outside the $c$-neighborhood of $v$. Then w.h.p. each connected component of $H$ has size $\bigO(\Delta^{2c}\log n)$.
\end{theorem}

Our shattering algorithm works as follows:

\medskip

\noindent\textbf{Shattering Algorithm:} \textit{Coloring phase:} Each node in $V$ colors itself red with probability $1/4$, blue with probability $1/4$ and remains uncolored otherwise.  
 \textit{Uncoloring phase:} Any $u\in U$ that has more than $3/4$ colored neighbors in $V$ uncolors all of its neighbors.

\medskip

After the execution of the shattering algorithm, a node $u\in U$ is \emph{satisfied} if it has at least one red and at least one blue neighbor, otherwise it is \emph{unsatisfied}.

\begin{lemma}\label{lem:shatteringProbability}
If $\Delta\geq c\log\rank$ for a sufficiently large constant $c$, the probability for a node $u\in U$ to be unsatisfied after the execution of the shattering algorithm is at most $e^{-\eta\Delta}$ for some $\eta>0$. In particular, the probability for $u$ to be unsatisfied is at most $(e\Delta\rank)^{-8}$.
\end{lemma}

\begin{proof}
If a node $u\in U$ is unsatisfied, one of the following events occurs after the coloring phase:
\begin{enumerate}
	\item $A_1(u)=$ less than $1/4$ of its neighbors are colored
	\item $A_2(u)=$ more than $3/4$ of its neighbors are colored
	\item $A_3(u)=$ between $1/4$ and $3/4$ of its neighbors are colored and the colored neighbors are all red or all blue
	\item $A_4(u)=$ there is a node $\bar{u}\in U$ two hops from $u$ which has less than $1/4$ or more than $3/4$ of its neighbors colored (in which case $\bar{u}$ uncolors a common neighbor and thus possibly destroys $u$'s satisfaction)
\end{enumerate}
Next, we bound the probability of the events $A_1(u),A_2(u), A_3(u)$ and $A_4(u)$. Let $d:=\text{deg}(u)$ and $v_1,\dots,v_d \in V$ be $u$'s neighbors. For every $i\in\{1,\dots,d\}$ introduce a random variable $X_i$ with $X_i=1$ if $v_i$ is colored and $X_i=0$ otherwise. Let $X:=\sum\nolimits_{i=1}^d X_i$. With Chernoff bounds and $\delta\geq \Delta/2$ we obtain 

\begin{align*}
\pr{A_1(u)} & =\pr{X<\frac{1}{4}d}\leq e^{-\frac{d}{16}}\leq e^{-\frac{\Delta}{32}}\\
\pr{A_2(u)} & \leq \pr{X>\frac{3}{4}d} \leq e^{-\frac{d}{16}}\leq  e^{-\frac{\Delta}{32}} \\
\intertext{With a union bound over $u$'s $2$-hop neighborhood $N^2(u)$ we get}
\pr{A_1(u)\vee A_2(u)\vee A_4(u)}& \leq\sum_{u'\in N^2(u)}\pr{A_1(u')\cup A_2(u')} \leq 2e^{-\frac{\Delta}{32}}\Delta \rank
\intertext{For $A_3$, given that at least $1/4$ of $u$'s neighbors are colored, we have $X\geq d/4$. Thus the probability that all colored neighbors are red (blue resp.) is at most $2^{-X}\leq 2^{-\frac{d}{4}}$. Hence, $\pr{A_3}\leq2\cdot2^{-\frac{d}{4}}\leq 2\cdot2^{-\frac{\Delta}{8}}$. 
	Thus we obtain }
\pr{u\text{ is unsatisfied}}& \leq\pr{A_1\vee A_2\vee A_3\vee A_4}\leq2e^{-\frac{\Delta}{32}}\Delta\rank+2\cdot2^{-\frac{\Delta}{8}}.
\end{align*}
With a sufficiently large constant $c$ in $\Delta\geq c\log\rank$ we obtain an $\eta>0$ such that $\pr{u\text{ is unsatisfied}}\leq e^{-\eta\Delta}\leq(e\Delta\rank)^{-8}$~.
\end{proof}

Based on \Cref{thm:shattering} and \Cref{lem:shatteringProbability}, we can now prove \Cref{thm:mainweaksplit}, our main randomized weak splitting result. \Cref{thm:shattering} and \Cref{lem:shatteringProbability} together can be used to show that after running the above shattering algorithm, the remaining problem is a weak splitting instance on components on $\poly\log n$ size. On these components, we can use our deterministic algorithm (\Cref{thm:detsplitting}) to solve the problem in $\poly\log\log n$ time.

\begin{proof}[\textbf{Proof of \Cref{thm:mainweaksplit}}]
Let $B=(U\cup V,E)$ be a bipartite graph. If $\delta>2\log n$, the trivial $0$-round algorithm (color each node in $V$ red/blue with probability $1/2$) computes a weak splitting w.h.p.: The probability of $u\in U$ to have a monochromatic neighborhood is at most $2\cdot2^{-\deg(u)}\leq2\cdot n^{-2}$. With a union bound over the nodes in $U$ we get that the failure probability of the algorithm is less than $2/n$. So in the following we assume $\delta\leq2\log n$.
Our algorithm first executes the shattering algorithm on $B$. The residual graph $H$ is the graph induced by the unsatisfied nodes in $U$ and the uncolored nodes in $V$ after the shattering. Then we use the deterministic weak splitting algorithm from  \Cref{thm:weakSplittingSpeedup} on the connected components of $H$. Let $n_H$ be the maximum size of a connected component of $H$. We need to show that $\delta_H\geq2\log n_H$ (to use the algorithm from \Cref{thm:weakSplittingSpeedup}) and $n_H=\text{poly}(\rank,\log n)$ (to achieve the stated runtime).

\smallskip

\textit{Upper Bounding $n_H$:} Let $G$ be the graph on the node set $U$ obtained by inserting an edge between two nodes in $U$ if they have a common neighbor in $B$. We have $\Delta_G\leq\Delta_B\cdot\rank_B$ (in the following let $\Delta=\Delta_B$ and $\rank=\rank_B$). 
To use \Cref{thm:shattering} on $G$ we define a randomized process on $G$ that is equivalent to the shattering algorithm: Each node $v\in V$ is assigned to its neighbor $u\in U$ with the smallest id. Then the behaviour of $v$, that is, picking a random color, informing neighbors about the color choice and potentially uncoloring itself is simulated by $u$. All other definitions remain the same, in particular, a node $u\in U$ is satisfied if there are $v',v''\in V$ that are simulated by nodes $u',u''\in U$ and $v'$ is colored red and $v''$ is colored blue. Thus, the process is a random process on $G$ in which a node $u\in V(G)$ is unsatisfied with probability at most $(e\Delta\rank)^{-8}$, even if random bits at nodes in distance larger than $2$ in $G$ are drawn adversarially.  With \Cref{thm:shattering} it follows that after the shattering, w.h.p., each connected component of the residual graph $H'\subseteq G$ induced by the unsatisfied nodes in $G$ has size $n_{H'}=\bigO(\Delta^4\rank^4\log n)$. As each node of $G$ has at most $\Delta=\bigO(\log n)$ neighbors we can bound the size of the connected components of unsatisfied nodes in $U$ and uncolored nodes in $V$ as $n_H\leq n_{H'}\cdot \Delta=O(\rank^4\log^6 n)$, w.h.p.

\smallskip

\textit{Lower Bounding $\delta_H$: } Due to the uncoloring phase of the shattering algorithm, each node has at least $1/4$ of its neighbors uncolored after the shattering. It follows that $\delta_H\geq\delta/4$. If we choose the constant $c$ in $\delta\geq c\log(\rank\log n)$ large enough, we get $\delta_H\geq2\log n_H$. 

\textit{Solving Small Connected Components/Runtime:} 
Due to $\delta_H\geq2\log n_H$ we can apply the deterministic algorithm from \Cref{thm:weakSplittingSpeedup} on the unsolved connected components of $H$.
The splitting algorithm including the uncoloring runs in $O(1)$ rounds. The application of \Cref{thm:weakSplittingSpeedup} on the small components has runtime  
\begin{align*}
&\bigO(\rank_H/\delta_H\cdot\log^2n_H+\log^3n_H(\log\log n_H)^{1.1})\\
&= \bigO\left(\frac{\rank}{\delta}\log^2(\rank\log n)+\log^3(\rank\log n)\left(\log\log (\rank\log n)\right)^{1.1}\right)~, 
\end{align*}
where we used  $n_H=\text{poly}(\rank,\log n)$, $R_H\leq R$ and $\delta_H\geq\delta/4$.
\end{proof}

\hide{
\begin{remark}
The exact runtime of \Cref{thm:mainweaksplit} is 
 \[\bigO\left(\frac{\rank}{\delta}\log^2(\rank\log n)+\log^3(\rank\log n)\left(\log\log (\rank\log n)\right)^{1.1}\right).\]
\end{remark}
}

\subsection{Lower Bound for Weak Splitting}


\begin{figure}[t]
  \begin{center}
	\begin{subfigure}{.24\linewidth}
	  \centering
	  \includegraphics[width=.85\linewidth]{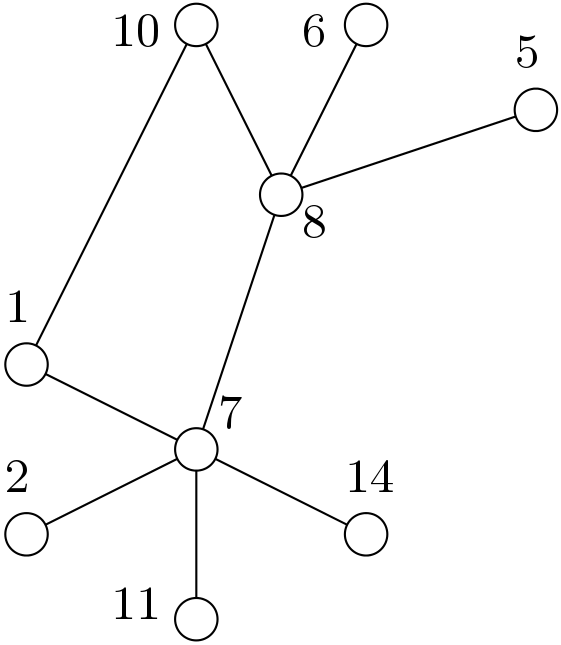} \caption{(A part of) the input graph}
	\end{subfigure}
	\begin{subfigure}{.24\linewidth}
	  \centering
	  \includegraphics[width=.85\linewidth]{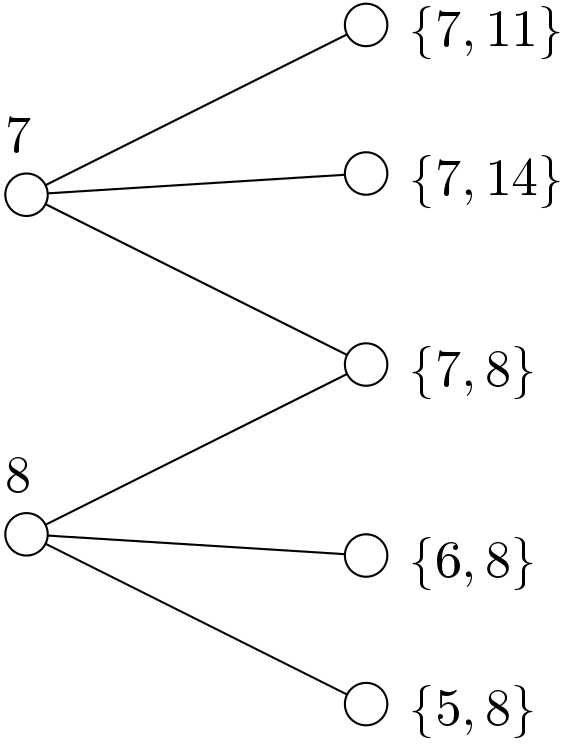} \caption{The bipartite graph $B$}
	\end{subfigure}
	\begin{subfigure}{.24\linewidth}
	  \centering
	  \includegraphics[width=.85\linewidth]{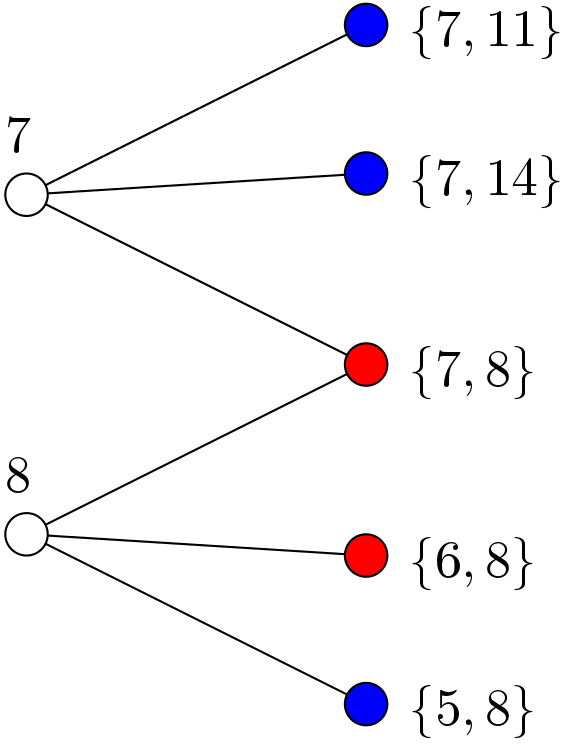} \caption{A weak splitting}
	\end{subfigure}
	\begin{subfigure}{.24\linewidth}
	  \centering
	  \includegraphics[width=.85\linewidth]{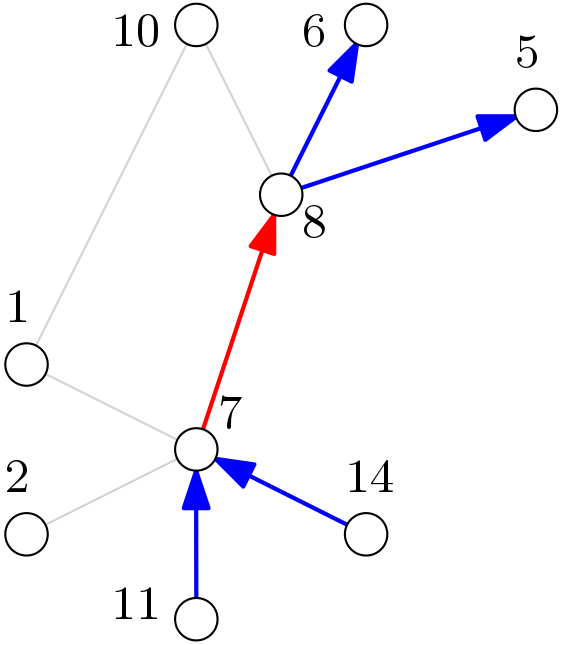} \caption{Resulting orientation}
	\end{subfigure}
	\caption{Figure (b) illustrates the neighborhoods of nodes $7$ and $8$ in the bipartite graph $B$. The graph $B$ has rank at most $2$, since the nodes on the right correspond to edges that have at most $2$ endpoints. Figure (c) shows a possible outcome of a weak splitting. Since the edge $\{ 7, 8\}$ is red, it gets directed from smaller to larger. Since $8$ has also a blue neighbor, at least one of its edges will get directed outwards.}
	\label{fig:weaksplit}
      \end{center}
\end{figure}

We conclude this section by giving a simple lower bound that shows that the weak splitting problem requires $\Omega(\log_{\Delta}\log n)$ time randomized and $\Omega(\log_\Delta n)$ time determinstically even on instances of rank $\rank=2$. The proof is based on a reduction from the sinkless orientation problem and the main steps of the reduction are depicted in \Cref{fig:weaksplit}. Given a graph $G$ on which we want to compute a sinkless orientation, we build a bipartite graph $B$ with right-hand side degree at most $2$ such that a weak splitting solution directly gives a sinkless orientation $G$. The lower bounds then follows from the sinkless orientation lower bounds in \cite{STOC16,CKP16}.

\begin{theorem}\label{lowerboundrand}
There is \emph{no} randomized distributed algorithm for solving the weak splitting problem in bipartite graphs with maximum degree $\Delta$, even if the rank is as small as $2$ in $o(\log_{\Delta}\log n)$ rounds.
\end{theorem}

\begin{proof}
We show how to reduce the sinkless orientation problem for graphs with minimum degree at least $5$ to the weak splitting problem. The statement then follows as there is an $\Omega(\log_{\Delta}\log n)$ lower bound for sinkless orientation (in regular graphs) from \cite{STOC16}.

For an illustration of the following process, we refer to \Cref{fig:weaksplit}.
We are given a graph $G=(V,E)$ with minimum degree $\delta_G\geq 5$ on which we want to compute a sinkless orientation by reduction to a weak splitting instance. For this purpose, we define a bipartite graph $B$ (i.e., a weak splitting instance) between the nodes and edges of $G$ in the following way: There is a left-hand-side node in $B$ for every node $v\in V$ of $G$ and there is a right-hand side node in $B$ for every edge $e\in E$ of $G$. We connect each left-hand-side node $u\in V$ to the right-hand-side nodes corresponding to some of $u$'s edges in $G$: if at least half of $u$'s neighbors in $G$ have a larger ID than $u$, we connect $u$ to every $e=\{u,v\}\in E$ for which $\text{ID}(v)>\text{ID}(u)$. Otherwise, i.e., if more than half of $u$'s neighbors have a smaller IDs than $u$, we connect $u$ to every $e=\{u,v\}\in E$ for which $\text{ID}(v)<\text{ID}(u)$. The resulting bipartite graph has rank at most $2$ and degree at least $\lceil\delta_G/2\rceil\geq 3$, which is sufficient for the weak splitting problem on $B$ to be solvable. A weak splitting solution on $B$ is a red/blue coloring of the right-hand-side nodes of $B$ and thus of the edges $E$ of $G$. Either (if at least half of $u$'s neighbors have a larger ID) $u$ gets both a red and a blue edge to a larger ID neighbor, or (if at more than half of $u$'s neighbors have a smaller ID), $u$ gets both a red and a blue edge to a lower ID neighbor. We therefore directly obtain a sinkless orientation of $G$ by orienting the edges $e\in E$ as follows. If an edge $e\in E$ is colored red, it gets directed from the node with smaller ID towards the node with higher ID, and if $e$ is colored blue, it is directed the other way around.

For the number of nodes in $B$, denoted by $n_B$, we have $n_B=|V|+|E|$ and for the degrees we have $\Delta_B\leq\Delta_G$. Thus, if one could solve a weak splitting in $o(\log_{\Delta_B}\log n_B)$ rounds, one could also compute a sinkless orientation in time $o(\log_{\Delta_G}\log n_G)$.
\end{proof}

\begin{corollary}\label{cor:detlower}
There does not exist a deterministic distributed algorithm solving the weak splitting problem in bipartite graphs in $o(\log_{\Delta}n)$ rounds.
\end{corollary}

\begin{proof}
If there was a $o(\log_{\Delta}n)$ deterministic algorithm for weak splitting, it could be transformed to run in $\bigO(\log^*n-\log^*\Delta+1)$ rounds (\cite{CKP16}), violating the lower bound from \Cref{lowerboundrand}.
\end{proof}

\section{Completeness of Weak Splitting Variants}
\label{sec:completeness}

In this section, we prove that the $(C,\lambda)$-multicolor splitting
and the $C$-weak multicolor splitting problems that were introduced in
\Cref{sec:contributions} are \PRLOCAL-complete. That is, if a
$\poly\log n$-time deterministic \LOCAL algorithm for one of the two
problems exists, then such an algorithm exists for all problems in
\PRLOCAL and thus in particular for problems such as
$(\Delta+1)$-coloring or MIS. In the two proofs, we use the \SLOCAL
model, which was introduced in \cite{GKM17}. The \SLOCAL model can be
seen as a sequential version of the \LOCAL model. In a
$\SLOCAL(t)$-algorithm, the nodes of a given graph $G$ are processed
in an arbitrary sequential order. Each node is assumed to have some
local memory, which initially just contains its unique ID and any
possible input to the problem we need to solve. When a node is
processed, it can read the current state of its $t$-hop neighborhood
and based on this information, it can store its output and possibly
additional information in its local memory. Before proving the
completeness of the two new problems, we give a slightly more precise
\PRLOCAL-completeness result than the one proved in \cite{GKM17}.

\begin{lemma}\label{lemma:weaksplittingcomplete}
  For bipartite graphs $B=(U\cup V, E)$ where all nodes $u\in U$ have degree $\deg(u)\geq 2\log n$, the weak splitting problem has a deterministic $\SLOCAL(2)$-algorithm and it is $\PRLOCAL$-complete.
\end{lemma}
\begin{proof}
  It follows directly from the proof of \cite[Theorem 1.8]{GKM17} that the given weak splitting problem is \PRLOCAL-hard. It remains to show that the problem is in \PRLOCAL. Consider the $0$-round randomized algorithm where each node in $V$ picks color red or blue uniformly at random. The probability that a node only sees one color is at most $2(1/2)^{2\log n}=2/n^2$. Since the validity of a weak splitting solution can be checked in a single round deterministically, it now follows directly from \cite[Theorem III.1]{GHK18} that the above weak splitting solution has a deterministic $\SLOCAL(2)$ algorithm. In \cite{GKM17}, it is shown that such an algorithm can be turned in to randomized \LOCAL model with $\poly\log n$ time complexity.
\end{proof}

\begin{theorem}\label{thm:weakmulticolorsplit}
  For bipartite graphs $B=(U \cup V, E)$ where all nodes in $u\in U$ have degree $\deg(u)\geq (2\log n+1)\ln^c n$ for some constant $c\geq 1$, the weak multicolor splitting problem is \PRLOCAL-complete for any $C\leq \polylog n$.
\end{theorem}

\begin{proof}
  We first show that for the given parameters, the problem is in
  \PRLOCAL. Consider the randomized process where every node in $V$
  chooses one of the first $\lceil 2\log n\rceil$ colors independently
  and uniform at random. For each node $u\in U$ and each color $x$
  (among the $\lceil 2\log n\rceil$ colors), the probability that no
  neighbor of $u$ chooses color $x$ is
  \[
  \left(1- \frac{1}{\lceil 2\log n\rceil}\right)^{\deg(u)}\  \stackrel{(c>1)}{<}\ 
  \left(1 - \frac{1}{2\log n + 1}\right)^{(2\log n +1)\ln n}\ < \ \frac{1}{n}.
  \]
  As the number of nodes in $U$ is less than $n$, it follows the
  expected number of nodes in $U$ that see less than $2\log n$
  different colors is less than $1$. As the above random process can
  be implemented in the \LOCAL model in $0$ rounds (i.e., without
  communication) and since the correctness of a weak multicolor
  splitting solution can be locally checked in a single round, it
  follows by \cite[Theorem III.1]{GHK18} that the above random process
  can be derandomized into a deterministic $\SLOCAL(2)$ algorithm that
  solves the given weak multicolor splitting problem. It is shown in
  \cite{GKM17} that such an algorithm can be turned into a
  $\poly\log n$-time randomized \LOCAL algorithm and we can therefore
  conclude that the given weak multicolor splitting problem is in
  \PRLOCAL.

  To prove the hardness of the problem, we reduce the weak splitting
  problem to the weak multicolor splitting problem in $\polylog n$
  rounds. Assume that we are given a bipartite graph $B=(U\cup V, E)$
  where all nodes in $u\in U$ have degree
  $\deg(u)\geq (2\log n+1)\ln^c n$ for some constant $c\geq 1$ on
  which we want to solve weak splitting. By
  \Cref{lemma:weaksplittingcomplete}, we know that weak splitting on
  bipartite graphs with those parameters is \PRLOCAL-complete. In
  order to solve weak splitting on $B$, we first solve the weak
  multicolor splitting with $C$ colors on $B$. Each node $u\in U$ is
  then guaranteed to have at least $2\log n$ (and thus also at least
  $\lceil 2\log n\rceil$) different colors among its neighbors. For
  $u\in U$, let $S(u)\subseteq R$ be a set of $\lceil2\log n\rceil$
  neighbors of $u$ such that all nodes in $S(u)$ have different
  colors. We transform the graph $B$ into a graph $B'$ by only keeping
  the $\lceil 2\log n\rceil$ edges for each $u \in U$ that connect $u$
  to its neighbors in $S(u)$. A valid weak splitting solution on $B'$
  is also a valid weak splitting solution on $B$ and we can therefore
  solve weak splitting on $B'$ instead of $B$. Since for each node
  $u\in U$, all its neighbors in $B'$ have different colors, any two
  nodes $v,w\in V$ at distance $2$ in $B'$ have different colors. The
  given coloring is therefore a proper partial $C$-coloring of the
  graph $B'^2$ in which each node in $V$ has a color. By using the
  method on \cite[Proposition III.2]{GHK18}, this coloring can be used
  to run an $\SLOCAL(2)$-algorithm on $B'$ in $O(C)$ rounds in the
  \LOCAL model as long as this \SLOCAL algorithm only needs to assign
  output values to the nodes in $V$ (i.e., to the colored nodes). By
  \Cref{lemma:weaksplittingcomplete}, the weak splitting problem on
  $B'$ has such an $\SLOCAL(2)$-algorithm and given the $C$-coloring
  the nodes in $V$, we can therefore compute a weak splitting of $B'$
  in $O(C)=O(\poly\log n)$ rounds deterministically in the \LOCAL
  model.
\end{proof}

\begin{theorem}\label{thm:multicolorsplit}
  Let the number of colors $C\in \mathbb{N}$ be such that $C\geq 2$
  and $C\leq \poly\log n$ and assume that
  $\lambda\geq \min\set{0.95,\frac{3}{C-1}}$. Then, for
  bipartite graphs $H=(L \cup R, E)$ where all nodes in $u\in L$ have
  degree $\deg(u)\geq \frac{\alpha}{\lambda}\cdot\ln^c(n)$
  for a sufficiently large constant $\alpha>0$ and some constant $c\geq 1$, $(C,\lambda)$-multicolor splitting is in
  \PRLOCAL and it is \PRLOCAL-complete if $\lambda \leq C^{-\eps}$ for
  some constant $\eps>0$ and if each node in $L$ has degree at least
  $\beta\ln^2n$ for a sufficiently large constant $\beta>0$.
\end{theorem}

\begin{proof}
  For $C=2$, the the \PRLOCAL-hardness follows directly from
  \Cref{lemma:weaksplittingcomplete} and since in this case,
  $\lambda\geq 0.95$, it is also straightforward to see that the
  problem is in \PRLOCAL if the minimum degree in $L$ is at least
  $\alpha\ln n$ for a sufficiently large constant $alpha>0$. We can
  therefore assume, w.l.o.g., that $C\geq 3$.  We first determine a
  number of color $C'$ as follows. If $\lambda\geq 2/3$, we choose
  $C'=3$, otherwise, we choose $C' := \lceil 3/\lambda\rceil$. Note
  that in both cases, we have $C'\leq C$. In the second case, this
  follows because we then have $\lambda \geq 3/(C-1)$ and thus
  $C'\leq 3/\lambda + 1 \leq C$.  To show that
  $(C,\lambda)$-multicolor splitting is in \PRLOCAL with the given
  parameters, consider the random process, where each node in $R$
  chooses one of $C'\leq C$ colors independently and uniformly at
  random. Let $u\in L$ be a node of degree
  $d\geq\frac{\alpha}{\lambda}\cdot\ln(n)$ (where $\alpha$ can later
  be chosen as a suitably large constant). We concentrate on one color
  $x$ of the $C'$ colors. Let $X\sim\mathrm{Bin}(d, 1/C')$ be the
  number of neighbors of $u$ that choose color $x$. We have
  \begin{equation}\label{eq:numcolorxbound}
  \Pr\left(X \geq  \lceil\lambda d\rceil \right) \leq
  {d \choose \lceil \lambda d\rceil}\cdot \frac{1}{C'^{\lceil\lambda
      d\rceil}}
  < \left(\frac{ed}{\lceil\lambda d\rceil C'}\right)^{\lceil\lambda d\rceil}
  \leq \left(\frac{e}{\lambda C'}\right)^{\lambda d}
  \leq \left(\frac{e}{\lambda C'}\right)^{\alpha\ln n}.
  \end{equation}
  The second inequality follows because ${n\choose k}<(en/k)^k$ for
  $k\geq 1$, The third inequality follows because for $x\geq 1$, $x^{-x}$ is a
  monotonically decreasing function. In order to show that a
  randomized algorithm exists, it is sufficient to show that for
  sufficiently large $\alpha>0$, the probability bound in
  \eqref{eq:numcolorxbound} is of the form $n^{-\Theta(\alpha)}$ and
  we thus need to show that $\frac{e}{\lambda C'}$ is a constant
  smaller than $1$.  Let us first consider the case, where $\beta\geq 0.95$. In this case,
  we have $C'=3$ and the claim follows because $e/(0.95 \cdot 3)
  =e/2.85 < 1$. Otherwise, we have $C'=\lceil 3/\lambda\rceil$ and
  thus $(e/\lambda C') \leq e/3$.

  It remains to prove that $(C,\lambda)$-multicolor splitting is
  \PRLOCAL-hard if $\lambda \leq C^{-\eps}$ and $c\geq 2$. We reduce
  from weak multicolor splitting on a bipartite graph $H=(L\cup R, E)$ to
  $(C,\lambda)$-multicolor splitting as follows. First note that if
  $\lambda \leq 1/(2\log n)$, a $(C,\lambda)$-multicolor splitting
  solution directly also solves weak multicolor splitting. If
  $\lambda > 1/(2\log n)$, our goal is to compute a
  $(C'',1/(2\log n))$-multicolor splitting on $H$ by using
  $\poly\log n$ instances of $(C,\lambda)$-multicolor
  splitting. Assume that the minimum degree any node $u\in L$ is at
  least $\beta \ln^2 n$ for a sufficiently large constant
  $\beta$. By \Cref{thm:weakmulticolorsplit}, we know that weak
  multicolor splitting is \PRLOCAL-complete for such graphs. The
  reduction consists of $\lceil\log_{1/\lambda}(2\log n)\rceil$
  iterations. We inductively prove that at the beginning of iteration $i\in \set{1,\dots,
    \lceil\log_{1/\lambda}(2\log n)\rceil}$, we are given a $(C^{i-1},
  \max{\lambda^{i-1}, 1/(2\log n)})$-multicolor splitting of $H$. The statment is clearly
  true for $i=1$ by just coloring each node in $R$ with a single
  color. For the $i^{\mathit{th}}$ iteration, each node $u\in L$,
  creates $C^{i-1}$ virtual nodes, one for each of the at most
  $C^{i-1}$ colors. The virtual node $u_x$ of $u$ corresponding to
  some color $x$ is connected to each neighbor $v\in R$ of $u$ that is
  colored with color $x$. We obtain the bipartite graph $H_i$ for the
  $(C,\lambda)$-multicolor splitting instance of iteration $i$ by
  taking the graph induced by the nodes in $R$ and the virtual nodes
  of degree at least $\frac{\alpha}{\lambda}\cdot\ln(n)$. Note that
  here, $n$ has to refer to the number of nodes of $H_i$, however
  since we can choose $\alpha$ sufficiently large, this is no
  problem. After running $(C,\lambda)$-splitting on $H_i$, each node
  in $R$ chooses a new color by combining its old color with the color
  computed in the current multicolor splitting instance. This results
  in a coloring with at most $C^i$ colors of the nodes in $R$. Because
  virtual nodes are split until their degree becomes
  $\frac{\alpha}{\lambda}\cdot\ln(n)$, after $i$ iterations, each node $u\in L$ has at most
  $\max\set{\lambda^i\cdot\deg(u), \frac{\alpha}{\lambda}\cdot\ln(n)}$
  neighbors of each color. Because we assume that $\deg(u)\geq
  \beta\ln^2 n$ for a
  sufficiently large constant $\beta$, this implies that we get a  $(C^i,
  \max{\lambda^i, 1/(2\log n)})$-multicolor splitting. This concludes
  the induction and it remains to prove that the total number of
  colors $C''=C^{\lceil \log_{1/\lambda}(2\log n)\rceil}$ is at most
  $\poly\log n$. However, this follows directly because we assumed
  that $\lambda \leq C^{-\eps}$ for some constant $\eps>0$.
\end{proof}

\section{Faster Splittings Imply Faster Coloring and MIS Algorithms}
\label{sec:parameter-preserving-reductions}
In this section, we explain that we can reduce the coloring problem
and also the MIS problem to the splitting problem, on (a subgraph of)
the same network. Hence, this reduction for instance preserves the
maximum degree of the network (or formally, it does not increase
it).

\newtheorem{observation}{Observation}

\newcommand{\TSP}{\ensuremath{T_{\textrm{USP}}}}
\newcommand{\TNSP}{\ensuremath{T_{\textrm{SP}}}}
\subsection{Vertex Coloring}

\paragraph{The Uniform Splitting Problem.}
Let $G = (V, E)$ be a graph with maximum degree $\Delta$ and minimum degree $\delta \geq \Delta/2$ and let $\eps > 0$.
In the splitting problem, the task is to divide the nodes in $V$ into two disjoint sets $R, B \subseteq V$.
The goal is that for each node $v$, the degree of the graphs induced by $R$ and $B$ is at most $(1/2 + \eps) d_G(v)$ and at least $(1/2 - \eps) d_G(v)$.
In the \emph{uniform} splitting problem, the input graph has a minimum degree of $\delta \geq \Delta / 2$.

\begin{remark}
	Consider the following slight modification of the uniform splitting problem.
	Instead of demanding an almost $\Delta$-regular graph, we may focus on general graphs and impose no restrictions on nodes of degree less than $\delta = \Delta/2$.
	It is clear that the uniform splitting problem is not easier than the modification.
	For a reduction from the modification to the original problem, consider a graph $G$ and add the following virtual gadgets to every node $v$ with $\deg(v) < \delta$.
	Construct a $\delta$-clique and add (virtual) edges from $\delta - \deg(v)$ nodes to $v$.
	The degree of $v$ becomes $\delta$ and the degrees of the virtual nodes are at most $\delta$.
	Then we can run a uniform splitting algorithm on the virtual graph and obtain a solution to the modified problem.
	The na\"{i}ve approach yields a graph of size $O(n \cdot \Delta)$, but the size can easily be reduced to $O(n)$.
	\label{remark: strong}
\end{remark}

\begin{lemma}\label{lemma:coloringreduction}
	Let $\mathcal{A}$ be an algorithm for the uniform splitting problem and let $\TSP(n, \eps)$ be its runtime.
	Then there is a $(1 + o(1))\Delta$-vertex-coloring algorithm with runtime $O(\log n \cdot \TSP(n, \eps) + \poly(\log n))$.
\end{lemma}
\begin{proof}
	Suppose that $\Delta = \omega(\log \log n)$.
	Otherwise, we can directly run a $(\Delta + 1)$ coloring algorithm in $O(\log \log n) = O(\log n \cdot \TSP(n, \eps)))$ time~\cite{fraigniaud16}.
	Set $\eps = 1 / \log^2 n$.
	We apply algorithm $\mathcal{A}$ recursively $r = \log \Delta - \log \log n$ times until the maximum degree drops to $\Delta^* = \poly \log n$.
	This takes $O(r \cdot \TSP(n, \eps)) = O(\log \Delta \cdot \TSP(n, \eps))$ time.

	We obtain $2^{r}$ subgraphs with maximum degree $\Delta^* = 2^{-r} \cdot (1 + \eps)^r \cdot \Delta$.
	Now, we color the subgraphs in with disjoint color palettes, in $\poly(\log n)$ time, using the algorithm of ~\cite{fraigniaud16}.
	Notice that $2^r = o(\Delta)$.
	In total, the number of colors we require is
\begin{align*}
	2^r \cdot (\Delta^* + 1) 	& = 2^{r} \cdot 2^{-r} \cdot (1 + \eps)^{r} \cdot \Delta + 2^r = (1 + \eps)^{r} \cdot \Delta  + o(\Delta)\\
								& \leq e^{\eps \cdot r} \cdot \Delta + o(\Delta) = e^{r / \log^2 n} \cdot \Delta + o(\Delta) = e^{1 / \Omega(\log n)} \cdot \Delta + o(1) \\
								& = \Delta + o(\Delta) + o(\Delta) = (1 + o(1))\Delta \qedhere
\end{align*}
\end{proof}

\subsection{Maximal Independent Set}

\begin{lemma}
	Let $\TNSP(n, \eps)$ be the runtime of a non-uniform strong splitting algorithm.
	Then there is an MIS algorithm with runtime $O(\log^4 n \cdot \log^2 \Delta \cdot \TNSP(n, \eps) + \poly(\log n))$.
	\label{lemma: misreduction}
\end{lemma}

\paragraph{The MIS Algorithm.}
Our MIS algorithm is divided into $O(\log \Delta)$ steps, where in each step we reduce the maximum degree by a factor of $2$.
Each of these steps consist of $O(\poly \log n)$ iterations of eliminating high degree nodes, where in each iteration, we reduce the number of high degree nodes by a polylogarithmic factor. Once the maximum degree is polylogarithmic, we can execute an MIS algorithm with runtime linear in the maximum degree on the remaining graph~\cite{barenboim2014distributed}.

\paragraph{Heavy Node Elimination.}
Consider a graph with maximum degree $\Delta$.
We call a node $v$ \emph{heavy}, if $\deg(v) \geq \Delta / 2$.
Let $G'$ be the graph induced by the heavy nodes and their neighbors.

We create a variable node for every node in $G'$ and a constraint node for each active node.
We connect the constraint node of a node $v$ to the variable nodes that correspond to the neighbors of $v$ in $G'$.
Then we use the splitting algorithm with $\eps = 1/\log^2 n$ to color the variable nodes red and blue.
All nodes whose variable node is blue become passive.
In addition, every node with fewer than $\log n$ red neighbors becomes passive.
The splitting step is repeated $2 \log \Delta - \log \log n - 1$ times to obtain a graph $G^*$ where all nodes have at most $\Delta \cdot (1/2 + \eps)^{\log \Delta - \log \log n - 1} < 4 \log n$ active neighbors and similarly, each heavy node has more than $\log n$ active neighbors.

We compute an MIS in $G^*$ and remove all the MIS nodes and their neighbors from the (original) graph.
This process is iterated until the set of heavy nodes is empty.

\begin{lemma}\label{lemma: MISsize}
	Let $G = (V, E)$ be a graph with maximum degree $\Delta$.
	Let $I \subseteq V$ be a maximal independent set.
	Then $|I| \geq |V| / (\Delta + 1)$.
\end{lemma}

\begin{proof}
	Consider the following process:
	Every node in the MIS gives one dollar to itself and its neighbors.	
	In total, at most $\ell = (\Delta + 1) \cdot |I|$ dollars are distributed.
	
	By definition of an MIS, every node gets at least one dollar.
	Hence, $\ell \geq n$ and furthermore, 
	\begin{equation*}
		 (\Delta + 1) \cdot |I| = \ell \geq n\quad\iff\quad |I| \geq \frac{n}{\Delta + 1} 		\qedhere
	\end{equation*}
\end{proof}

\begin{lemma}\label{lemma: heaviesdie}
	Let $I$ be an MIS on $G^*$.
	Then at least $\Omega(|V_H| / \log^3 n)$ heavy nodes are covered by $I$.
\end{lemma}

\begin{proof}
	Let $A$ be the graph induced by the active nodes.
	By the design of our algorithm, every active node has degree at most $4\log n$.
	We make the following observations.
	\begin{enumerate}
		\item At least $|A| / (5 \log n)$ nodes are selected to $I$. This follows by \cref{lemma: MISsize}.
		\item Every active node is either heavy or has at least one heavy neighbor. This is true by definition of $G'$.
	\end{enumerate}
	Combining the observations with the maximum degree of $A$, we get that at least $|A| / (4\log n \cdot 5\log n) = |A| / (20\log^2 n)$ heavy nodes are neighbored by a node in $I$.
	Furthermore, combining the observations with the fact the every heavy node has at most $4\log n$ active neighbors, we get that $|V_H| \leq |A| \cdot 4\log n$.
	Putting all of the above together, we have that at least 
	\[
		\frac{|A|}{20\log^2 n} \geq \frac{|V_H|}{4\log n} \cdot \frac{1}{ 20\log^2 n} = \frac{|V_H|}{80\log^3 n}
	\]
	heavy nodes have a neighbor in $I$.
\end{proof}

\begin{proof}[Proof of \cref{lemma: misreduction}]
	In every iteration of the heavy node elimination method, we perform at most $O(\log \Delta)$ degree splittings.
	Furthermore, by \cref{lemma: heaviesdie} and by observing that always at least one heavy node is eliminated, after $O(\log^4 n)$ repetitions, all the heavy nodes are eliminated.
	Our algorithm consists of $O(\log \Delta)$ executions of the heavy node eliminating, plus some $\poly(\log n)$ time to handle the resulting graphs with $\poly(\log n)$ degrees, hence resulting in a total runtime of $O(\log^4 n \cdot \log^2 \Delta \cdot \TNSP(n, \eps) + \poly(\log n))$.
\end{proof}

\section{Weak Splitting in High Girth Graphs}

We recall the shattering algorithm from \Cref{sec:rand}.

\medskip

\noindent\textbf{Shattering Algorithm:}  \textit{Coloring phase:} Each node in $V$ colors itself red with probability $1/4$, blue with probability $1/4$ and remains uncolored otherwise.  
\textit{Uncoloring phase:} Any $u\in U$ that has more than $3/4$ colored neighbors in $V$ uncolors all of its neighbors.

\begin{lemma}\label{lem:gap}
For bipartite graphs $B$ of girth at least $10$ with $\delta\geq c\sqrt{\ln n}$ and $\Delta\geq c'\ln r$ for sufficiently large constants $c$ and $c'$, after running the shattering algorithm on $B$, the following holds: The graph $H$ induced by the unsatisfied nodes in $U$ and the uncolored nodes in $V$ after the shattering has $\delta_{H}\geq6\cdot\rank_{H}$, w.h.p.
\end{lemma}

\begin{proof}
Let $v\in V$. For a neighbor $u$ of $v\in V$, let $E_u$ be the event that $u$ is satisfied under the condition that $v$ remained uncolored. $E_u$ only depends on the random bits of nodes $\neq v$ within $u$'s $3$-hop neighborhood (the random bits of nodes that are $3$ hops away from $u$ can cause a node that is $2$ hops away from $u$ to uncolor a neighbor of $u$). Two neighbors $u$ and $\bar{u}$ of $v$ do not have a common node $\bar{v}\neq v$ that is within both $u$'s and $\bar{u}$'s $3$-hop neighborhood as otherwise we would have a cycle of size at most $8$ contradicting that the graph has girth at least $10$. It follows that the events $E_u$ and $E_{\bar{u}}$ are independent (when conditioning on the event that $v$ remains uncolored). For a neighbor $u$ of $v$ let $X_u$ be the random variable with $X_u=1$ if $u$ is unsatisfied and $X_u=0$ otherwise and let $X:=\sum\nolimits_{u\in\NH(v)}X_u$. By choosing $c'$ in $\Delta\geq c'\ln r$ sufficiently large we get $\pr{\text{$u$ is unsatisfied}}\leq e^{-\eta\delta}$ for some $\eta>0$ (\Cref{lem:shatteringProbability}) and $\eta\delta\geq2\ln r+2$. We choose the $c$ in $\delta\geq c\sqrt{\ln n}$ such that $c\geq\sqrt{96/\eta}$. We get \[\pr{X\geq k}\leq\binom{\rank}{k}e^{-\eta\delta k}\leq\left(\frac{e\rank}{k}\right)^ke^{-\eta\delta k}\leq e^{-k(\delta\eta-\ln r-1)}\leq e^{-\frac{k\eta\delta}{2}}\] where the last inequality holds because $\eta\delta\geq2\ln r+2$. For $k=\delta/24$ we have \[\pr{X\geq\delta/24}\leq e^{-\frac{\delta^2\eta}{48}}\leq e^{-\frac{c^2\eta\ln n}{48}}\leq e^{-\frac{96\ln n}{48}}=n^{-2}~.\] It follows that $\rank_{H}\leq\delta/24$, w.h.p. By the construction of the shattering algorithm we have $\delta_{H}\geq\delta/4$ and hence $\delta_{H}\geq6\cdot\rank_{H}$.
\end{proof}

\begin{theorem}\label{thm:girth10det}
For bipartite graphs $B$ of girth at least $10$ with $\delta\geq c\sqrt{\ln n}$ and $\Delta\geq c'\ln r$ for sufficiently large constants $c$ and $c'$, there is a deterministic algorithm that solves the weak splitting problem in $\bigO(\Delta^2\rank^2+\polylog n)$ rounds.
\end{theorem}

\begin{proof}
The shattering algorithm is a $1$-round randomized algorithm with checking radius one (degree and rank of a node are locally checkable with radius one). By \cite[Theorem III.1]{GHK18}, this algorithm can be derandomized into an $\SLOCAL(4)$-algorithm. By \cite[Proposition 3.2]{GHKfull}, this can be transformed into an $O(C)$-\LOCAL algorithm if a $C$-coloring of $B^4$ is given. As the maximum degree of $B^4$ is $\Delta^2\rank^2$ we can compute the necessary coloring with $\bigO(\Delta^2\rank^2)$ colors and in $\bigO(\Delta^2\rank^2+\log^*n)$ rounds. Thus the runtime for the derandomization can be bounded by $\bigO(\Delta^2\rank^2+\log^*n)=O(\Delta^2\rank^2)$ as $\Delta\geq \delta\geq c\sqrt{\ln n}$.

After we obtained a subgraph $H$ of $B$ with $\delta_{H}\geq6\cdot\rank_{H}$, we can solve a weak splitting on $H$ in $\polylog n$ rounds (\Cref{thm:DeltaLargersixr}).
\end{proof}

\begin{theorem}\label{thm:girth10rand}
For bipartite graphs $B$ of girth at least $10$ with $\delta\geq c\sqrt{\ln(\Delta\rank\ln n)}$ and $\Delta\geq c'\ln r$ for sufficiently large constants $c$ and $c'$, there is a randomized algorithm that solves the weak splitting problem in $\bigO(\Delta^2\rank^2+\polylog(\Delta\rank\log n))$ rounds, w.h.p.
\end{theorem}

\begin{proof}
We first run the splitting algorithm on $B$. The graph $H$ induced by the unsatisfied nodes in $U$ and the uncolored nodes in $V$ after the shattering has connected components of maximum size $n_H=\poly(\Delta,\rank,\log n)$, w.h.p. (cf. proof of \Cref{thm:mainweaksplit}). By the construction of the shattering algorithm we have $\delta_H\geq\delta/4\geq\frac{c}{4}\sqrt{\ln (\Delta\rank\ln n)}=\bar{c}\sqrt{\ln n_H}$ with $\bar{c}$ sufficiently large if $c$ was chosen sufficiently large. We also have $\Delta_H\geq\Delta/4\geq\frac{c'}{4}\ln r\geq c''\ln r_H$ with $c''$ sufficiently large if $c'$ was chosen sufficiently large. Therefore, we can apply the deterministic algorithm from \Cref{thm:girth10det} with runtime $\bigO(\Delta_H^2\rank_H^2+\polylog n_H)=\bigO(\Delta^2\rank^2+\polylog(\Delta\rank\log n))$ on these components.
\end{proof}

\clearpage
\bibliographystyle{alpha}
\bibliography{references}

\end{document}